%% file: main.tex
\documentclass[journal]{IEEEtran}

\input{utils/notation}

\newif\ifarxiv
\arxivtrue

\newif\ifincludeappendices
\includeappendicestrue

\begin{document}

\title{Adaptive Calibration of Day-Ahead Flexibility Reserve Requirements Under Distribution Shift}

\author{Napoleon~Costilla-Enriquez,~\IEEEmembership{Member,~IEEE},
    Yang~Weng,~\IEEEmembership{Senior~Member,~IEEE},
    and~Miguel~A.~Ortega-Vazquez,~\IEEEmembership{Senior~Member,~IEEE}
    \thanks{N. Costilla-Enriquez and M. A. Ortega-Vazquez are with the
    Electric Power Research Institute, Palo Alto, CA, USA (e-mail:
    ncostilla-enriquez@epri.com; maov@epri.com).}
    \thanks{Y. Weng is with Arizona State University, Tempe, AZ, USA
    (e-mail: yweng2@asu.edu).}\vspace{-10mm}}

\ifarxiv
    \onecolumn
    \thispagestyle{empty}
    \noindent
    This work has been submitted to the IEEE for possible publication.
    Copyright may be transferred without notice, after which this version
    may no longer be accessible.
    \vspace*{\fill}
    \setcounter{page}{0}
    \twocolumn
\fi

\maketitle

\input{Sections/00_abstract}
\input{Sections/01_introduction}
\input{Sections/02_background}
\input{Sections/03_online_cqr}
\input{Sections/04_online_crc}
\input{Sections/05_experimental_design}
\input{Sections/06_results}
\input{Sections/07_conclusion}

\bibliographystyle{IEEEtran}
\bibliography{utils/bibliography}

\ifincludeappendices
    \clearpage
    \appendices
    \raggedbottom
    \input{Appendices/proofs}
\fi

\end{document}

%% file: utils/notation.tex
\usepackage{amsmath,amssymb,amsthm}
\usepackage{bm}
\usepackage{bbm}
\usepackage{graphicx}
\usepackage{booktabs}
\usepackage[caption=false,font=footnotesize]{subfig}
\usepackage[table]{xcolor}
\usepackage{float}
\usepackage{placeins}
\usepackage{cite}
\usepackage[obeyspaces]{xurl}
\usepackage[hidelinks]{hyperref}
\usepackage[capitalise,nameinlink]{cleveref}

\newtheorem{theorem}{Theorem}
\providecommand{\Description}[1]{}

\newcommand{\indicator}[1]{\mathbbm{1}\!\left\{#1\right\}}
\newcommand{\E}{\mathbb{E}}
\newcommand{\Prob}{\mathbb{P}}
\DeclareMathAlphabet{\mathdutchcal}{U}{dutchcal}{m}{n}
\SetMathAlphabet{\mathdutchcal}{bold}{U}{dutchcal}{b}{n}
\newcommand{\netload}{\mathdutchcal{n}}

\newcommand{\cG}{\mathcal{G}}
\newcommand{\cH}{\mathcal{H}}

\newcommand{\cR}{\mathcal{R}}

\newcommand{\sR}{\mathbb{R}}

\crefname{equation}{Eq.}{Eqs.}
\Crefname{equation}{Equation}{Equations}
\crefname{figure}{Fig.}{Figs.}
\Crefname{figure}{Figure}{Figures}
\crefname{table}{Table}{Tables}
\Crefname{table}{Table}{Tables}

%% file: Sections/00_abstract.tex
\begin{abstract}
The increasing integration of variable renewable generation introduces
weather-driven variability and uncertainty into power-system operations. System
operators manage the resulting deviations between day-ahead forecasts and
real-time net load by scheduling upward and downward flexibility reserves.
These requirements have traditionally been set using simple deterministic
rules. Newer data-driven methods tailor reserves to forecasts and operating
conditions. However, their reliability may deteriorate when conditions depart
from the historical data. We address this
limitation with a sequential framework
that updates day-ahead reserve requirements as operating outcomes become
available. Within this framework, Adaptive Conformal Inference targets the
frequency of insufficient reserve, whereas Rolling Risk Control targets
expected uncovered deviation, the mean residual beyond the scheduled reserve.
We prove finite-horizon bounds on realized miscoverage and capped
uncovered-deviation loss. We
test the framework on real ERCOT data in a chronological deployment simulation
that respects operational timing: outcomes from one day first affect the
reserve schedule two days later. Training and validation use 2022--2023 data;
the uninterrupted test covers 2024--2025.
Coverage calibration achieves approximately 90\% coverage, with upward and
downward miscoverage near 5\%. Risk calibration brings expected uncovered
deviations close to their budgets while using 5.3\% less reserve than the
uncalibrated requirements. The framework thus produces adaptive reserve
requirements that provably control the selected reliability metrics over finite
horizons without assuming stationarity or exchangeability.
\end{abstract}


%% file: Sections/01_introduction.tex
\section{Introduction}
\label{sec:introduction}

Variable renewable energy sources~(VRES), energy storage, distributed energy
resources, electrification, demand-side flexibility, and data-center demand
are reshaping power-system operations~\cite{nerc2025ltra}. Together, these
changes alter net-load \emph{uncertainty} and
\emph{variability}~\cite{ortega2020risk,ortega2022generation}. Day-ahead
operating decisions rely on forecasts, but realized net load can differ from
the scheduled forecast because of forecast errors, fluctuations within the
operating interval, or both. Throughout this paper, \emph{net-load deviation}
denotes this forecast-to-realization difference. Its distribution depends on
weather, hour, season, load, and renewable production, and can evolve with the
power system and its forecasting practices. \Cref{fig:forecast_errors}
illustrates this behavior in ERCOT through seasonal changes in dispersion and
interannual changes in the mean and tails. To manage these deviations,
operators schedule
flexibility reserves\footnote{In this paper, flexibility reserves denote the upward and
downward capacity scheduled day ahead to cover continuous
forecast-to-realization net-load deviations over each operating interval. They
are distinct from contingency reserves for discrete outages and regulation
reserves for faster balancing services.} to balance supply and demand while the
deviation distribution changes
\cite{kirschen2018fundamentals,holttinen2012methodologies,
costilla2023operating}.
When scheduled reserve is insufficient, the residual imbalance must be managed
through faster corrective actions; severe shortages can contribute to
under-frequency and could even require load shedding, while severe surpluses
can increase curtailment and over-frequency risk.

Reserve products and procurement rules vary across system operators, including
ERCOT, CAISO, and NYISO~\cite{ercotASMethodology2026,caisoBPM2026,
nyisoASManual2026}. The methods used to set reserve requirements range from
deterministic rules to probabilistic calculations and data-driven approaches.
Deterministic rules set fixed requirements or formulas tied to a few system
quantities~\cite{ela2011operating,holttinen2012methodologies}. They are
transparent and easy to apply but cannot reflect all conditions that shape
net-load deviations. Probabilistic and data-driven studies estimate deviation
distributions, prediction intervals, or imbalance risk and select reserve
requirements to meet reliability targets~\cite{matos2011setting,
bessa2012reserve,costilla2023operating,zhao2021operating,
fahiman2019dynamic,devos2019dynamic}. Estimating a complete conditional
distribution becomes difficult when many operating variables are considered
because few historical observations resemble a particular combination of
weather, hour, load, and renewable production~\cite{bishop2006pattern}.
Quantile models address this limitation by estimating the required conditional
quantiles directly from richer inputs, using methods such as gradient-boosted
trees, neural networks, and quantile regression forests
\cite{chen2016xgboost,cannon2011quantile,meinshausen2006quantile}. NESO's
Dynamic Reserve Setting project illustrates this approach with
gradient-boosted quantile models that use operational and weather data to
recommend condition-dependent day-ahead reserve levels~\cite{nesoDRS2020}.

Quantile models tailor reserve requirements to forecast operating conditions,
but nominal quantiles do not provide a finite-sample guarantee of realized
coverage. Coverage is the probability or observed frequency with which the
scheduled interval contains the realized deviation. Changes in the deviation
distribution compound this limitation: the seasonal and interannual
differences in~\cref{fig:forecast_errors} show why a calibration that is
reliable in one period can become outdated in another. Periodic retraining
updates the predictions but does not explicitly correct deviations from the
coverage target. A separate mechanism is therefore needed to calibrate reserve
requirements from their observed performance.

\begin{figure*}[t]
    \centering
    \includegraphics[width=0.94\textwidth]{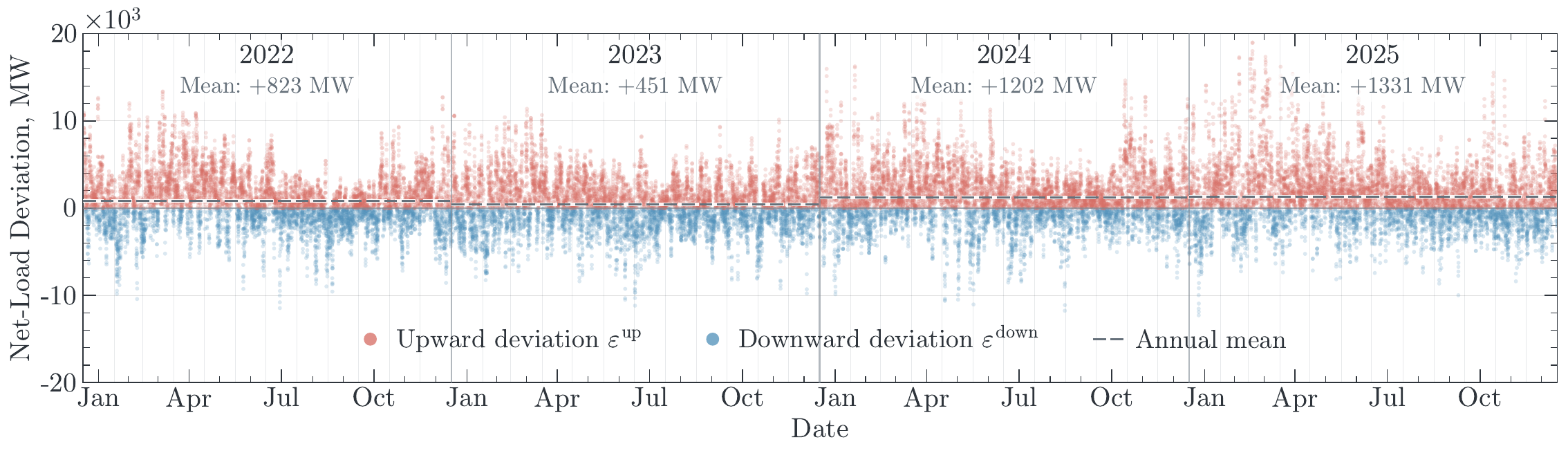}
    \caption{ERCOT net-load deviations,
    2022--2025. Positive and negative values require upward and downward
    reserve, respectively; dashed segments show annual means. Seasonal changes
    in dispersion and interannual changes in the mean and tails illustrate
    heteroscedasticity and distribution shift.}
    \Description{A time-series scatter plot shows positive and negative ERCOT
    forecast-to-realization net-load deviations from 2022 through 2025, with
    annual mean segments
    and vertical year boundaries.}
    \label{fig:forecast_errors}
\end{figure*}

Conformalized Quantile Regression~(CQR) adjusts fitted conditional quantiles to
provide finite-sample marginal coverage under exchangeability\footnote{A
collection of observations is exchangeable if its joint distribution is
unchanged when the observations are reordered.}
\cite{romano2019conformalized}. The calibration can be applied to any fitted
quantile model without a parametric deviation distribution.
Previous conformal applications in power systems include electricity-price and
renewable-generation forecasting~\cite{kath2021conformal,dhaka2026physics}.
Applying conformal calibration to reserve sizing requires translating observed
reserve performance into subsequent upward and downward requirements.
Here, the calibrated interval is an operational input: its endpoints become
physical upward and downward reserve requirements, and failures in the two
directions can have different operational consequences.

Coverage is important but incomplete: a small exceedance and a severe one
contribute equally to miscoverage, although they leave very different
imbalances for the operator to manage. We therefore also consider \emph{expected uncovered
deviation}~(EUD), the expected amount of deviation left uncovered after the
scheduled upward or downward reserve is exhausted. Conformal Risk Control~(CRC)
complements CQR by calibrating reserve size to bound this loss
\cite{angelopoulos2024conformal}. In their standard offline form, CQR and CRC
use a fixed historical calibration, and their guarantees rely on exchangeability
between calibration data and future observations~\cite{barber2023beyond}. In
power systems, however,
the resource mix and forecasting practices evolve, changing operating
conditions and the distribution of net-load deviations. A fixed calibration
can therefore become outdated, so coverage or EUD may depart from its target.

We therefore formulate reserve sizing as a sequential calibration problem. The
goal is to update condition-dependent upward and downward requirements from
observed performance so that miscoverage and EUD remain aligned with their
targets as the deviation distribution changes. This formulation distinguishes
how often reserve is exceeded from how much deviation remains uncovered,
treats the two exposure directions separately, and uses each outcome only after
it becomes available.

Our framework combines condition-dependent reserve estimates with adaptive
calibration to control realized miscoverage and EUD as conditions change.
Adaptive Conformal Inference~(ACI) calibrates CQR for miscoverage
\cite{gibbs2021adaptive}, whereas Rolling Risk Control~(Rolling RC) calibrates
reserve requirements for uncovered-deviation loss
\cite{feldman2023rolling}. We implement both methods chronologically, using
only the information available when each schedule is formed, and prove
finite-horizon bounds on realized time-average miscoverage and capped EUD loss
without assuming stationarity or exchangeability.

Although we present the proposed approach for hourly day-ahead-to-real-time
deviations, it is flexible and can be adapted to other market configurations,
time horizons, scheduling practices, and feedback delays imposed by data
availability.

The main contributions are:
\begin{enumerate}
    \item We formulate sequential reserve sizing for miscoverage and EUD,
    mapping adaptive states into deterministic upward and downward
    requirements. These requirements can be integrated directly into existing
    unit-commitment and production-cost models without changing their
    formulation or requiring scenario-based stochastic optimization
    \cite{kirschen2018fundamentals}. The formulation also respects operational
    timing: each reserve schedule is calibrated using only the information
    available when it is formed.
    \item We apply ACI-CQR to adaptive reserve sizing and introduce Directional
    ACI-CQR, with separate lower- and upper-tail updates. We prove
    finite-horizon identities and bounds for total and directional
    miscoverage without stationarity or exchangeability.
    \item We develop directional Rolling RC schemes for separate EUD budgets.
    We prove finite-horizon bounds on capped uncovered-deviation loss by
    direction and calibration group.
    \item On an uninterrupted 2024--2025 ERCOT test sequence, with all choices
    fixed using 2022--2023 data, we quantify the reserve cost of correcting
    directional imbalance and controlling reliability separately by hour under
    seasonal and interannual distribution changes.
\end{enumerate}

The rest of the paper is organized as follows. \Cref{sec:background}
formulates the reserve-sizing problem. \Cref{sec:cqr,sec:crc} develop the
adaptive calibration methods, \cref{sec:experiments} presents their empirical
evaluation, and \cref{sec:conclusion} concludes the paper.

%% file: Sections/02_background.tex
\section{Background and Problem Formulation}
\label{sec:background}

\subsection{Problem Formulation and Operational Timing}
\label{subsec:problem}

System operators procure upward and downward flexibility reserves to manage
deviations between forecast and realized net load. Changes in the resource and
demand mix can alter the magnitude and distribution of these deviations over time.
Depending on the scheduling horizon and temporal resolution, a net-load
deviation may reflect forecast uncertainty, intra-interval variability, or
both~\cite{ortega2022generation}.

Let $d\in\{1,\ldots,D\}$ index operating days and let
$h\in\cH:=\{0,\ldots,H-1\}$ index the operating intervals within a day. Let
$\netload^{\mathrm{F}}_{d,h}$ denote the net-load forecast for interval $(d,h)$
(e.g., a day-ahead forecast), and let
$\netload^{\mathrm{A}}_{d,h}$ denote the real-time realization of net load. The
net-load deviation is
\begin{equation}
    \varepsilon_{d,h}
    :=\netload^{\mathrm{A}}_{d,h}-\netload^{\mathrm{F}}_{d,h}.
    \label{eq:net_load_deviation}
\end{equation}
A positive deviation represents a shortage, or short exposure, and requires
upward reserve. If it is not covered, the system may experience under-frequency
or negative area control error and could even require load shedding. A negative
deviation represents a surplus, or long exposure, and requires downward
reserve. If it is not covered, the system may experience over-frequency,
positive area control error, or excessive VRES
curtailment~\cite{ortega2020risk}.
We index the two exposure directions by
$\mathcal J:=\{\mathrm{short},\mathrm{long}\}$.

The net-load deviation $\varepsilon_{d,h}$ captures the combined effect of
deviations in load, wind, and solar generation. Its conditional distribution
can therefore depend on the forecast levels of these components, hour of day,
season, weather, and other predictors. We denote these conditioning variables
by the feature vector $\bm{\phi}_{d,h}\in\sR^p$. These dependencies can also change
over time.

For each interval $(d,h)$, the operator uses the current feature vector
$\bm{\phi}_{d,h}$ and all previously revealed outcomes to select nonnegative
upward and downward reserves
$\bigl(r^{\mathrm{up}}_{d,h},r^{\mathrm{down}}_{d,h}\bigr)
\in\sR_{\geq0}^{2}$ before $\varepsilon_{d,h}$ is observed. The reserve pair
is selected to cover the net-load deviation with the desired reliability,
specified through either a coverage target or an EUD budget.

Reserve decisions must follow the timing and information available to the
system operator in practice. With interval-by-interval scheduling, the outcome
from interval $(d,h)$ may inform a later reserve decision. With day-ahead
scheduling, all $H$ reserve pairs for day $d$ are selected before any outcome
from that day is observed, using forecasts, calendar features, and earlier
outcomes. Our theory accommodates these and other operational scheduling
arrangements by requiring each decision to use only the information available
when it is made. In the ERCOT implementation studied here, outcomes from day
$d$ are incorporated only after the operating day is complete. By then, the
day-ahead scheduling cutoff for day $d+1$ has passed, so these outcomes first
affect the reserve schedule for day $d+2$.

Day-ahead reserve scheduling is common in U.S. organized markets, although the
products, uncertainty horizons, and procurement rules vary by operator. ERCOT
establishes ancillary-service requirements by service and hour
before the Operating Day~\cite{ercotNodalProtocols2026}; CAISO procures
bidirectional Imbalance Reserves in the day-ahead market to address
day-ahead-to-real-time net-load uncertainty~\cite{caisoBPM2026}; and NYISO
co-optimizes 10- and 30-minute operating reserves in its Day-Ahead and
Real-Time Markets~\cite{nyisoASManual2026}. Together, these examples show that
advance reserve scheduling is a common operational practice implemented
through market-specific products and time scales.

Our framework can accommodate these market-specific configurations by allowing
each operator to select its scheduling horizon and temporal resolution and to
define the corresponding forecast, realization, reserve decision, and feedback
sequence. Each update then uses only the outcomes available when the affected
decision is made, and the feedback lag follows the operator's scheduling and
data-availability rules.
We use hourly day-ahead-to-real-time deviations as a representative baseline
because they provide a common, operationally relevant advance-scheduling
setting.

\subsection{Coverage Reliability}
\label{subsec:coverage}

Let $\varepsilon$ be a real-valued random variable representing the net-load
deviation and let $R^{\mathrm{up}},R^{\mathrm{down}}\geq0$ be the scheduled
reserves. For the reserve interval
$\cR:=[-R^{\mathrm{down}},R^{\mathrm{up}}]$, the coverage target is
\begin{equation}
    \Prob\!\left\{\varepsilon\in\cR\right\}\geq1-\alpha,
    \label{eq:coverage_target}
\end{equation}
where $\alpha$ is the target miscoverage; $\alpha=0.1$ corresponds to
90\% coverage. For a realized deviation $\varepsilon_{d,h}$, the directional
miss indicators are
$M^{\mathrm{short}}_{d,h}:=
\indicator{\varepsilon_{d,h}>r^{\mathrm{up}}_{d,h}}$ and
$M^{\mathrm{long}}_{d,h}:=
\indicator{\varepsilon_{d,h}<-r^{\mathrm{down}}_{d,h}}$.
These events are disjoint, so the total miss indicator is
$M_{d,h}:=M^{\mathrm{short}}_{d,h}+M^{\mathrm{long}}_{d,h}$.
Separate directional targets may satisfy
$\alpha^{\mathrm{short}}+\alpha^{\mathrm{long}}=\alpha$. Unless stated
otherwise, we use the symmetric allocation
$\alpha^{\mathrm{short}}=\alpha^{\mathrm{long}}=\alpha/2$.

\subsection{Expected Uncovered Deviation}
\label{subsec:eud_metric}

Coverage counts misses but not their magnitude. For the same random variable
$\varepsilon$ and reserves, expected uncovered deviation~(EUD) is
$\E[(\varepsilon-R^{\mathrm{up}})^+]$ for short exposures and
$\E[(-\varepsilon-R^{\mathrm{down}})^+]$ for long exposures,
where $(x)^+:=\max\{0,x\}$. EUD is the risk measure used here: it has the same
units as the deviation, is non-increasing in reserve, and reflects both
how often reserve is exceeded and how much deviation remains uncovered. For a
realized deviation, the corresponding losses are
$\ell^{\mathrm{short}}_{d,h}:=
(\varepsilon_{d,h}-r^{\mathrm{up}}_{d,h})^+$ and
$\ell^{\mathrm{long}}_{d,h}:=
(-\varepsilon_{d,h}-r^{\mathrm{down}}_{d,h})^+$.
The operator selects budgets
$\rho^{\mathrm{short}},\rho^{\mathrm{long}}>0$, which need not be equal because
short and long exposures can have different consequences.

\subsection{Offline Conformal Calibration}

Let $\widehat f(\bm{\phi})$ be a fitted quantile model that produces lower and
upper estimates of the net-load deviation. Conformalized Quantile
Regression~(CQR) provides a calibration layer on top of this model. The model is
first fitted on training data. CQR then evaluates its interval on held-out
calibration observations, measures how far those deviations fall outside the
interval, and adjusts the interval to attain the requested coverage
\cite{romano2019conformalized}. Its finite-sample marginal coverage guarantee is
distribution-free and valid for any fitted quantile model.

For EUD, a fitted model can likewise provide a family of condition-dependent
reserve estimates indexed by a parameter that controls reserve size and,
therefore, uncovered-deviation loss. Conformal Risk Control~(CRC) uses a
held-out calibration set to select this parameter so that the expected bounded
loss derived from EUD satisfies a specified budget
\cite{angelopoulos2024conformal}. The reserve family may be constructed from any
fitted model; in our framework, conditional quantile models provide the
condition-dependent base reserves. When the loss is bounded and non-increasing
as the calibrated reserve grows, the resulting expected-loss guarantee is
distribution-free.

Both methods may also be calibrated separately within prespecified groups,
yielding group-specific coverage or expected-loss guarantees
\cite{vovk2012conditional,angelopoulos2024conformal}. In the offline setting,
these guarantees require exchangeability between the fixed calibration set and
future observations within each group. This assumption is difficult to justify
in reserve scheduling because the distribution of net-load deviations can
change over time, as illustrated in \cref{fig:forecast_errors}. At the same
time, each realized deviation provides feedback for later reserve decisions.
This sequential feedback motivates the adaptive CQR and CRC schemes developed
next. They regulate realized time-average miscoverage or capped EUD loss
without requiring exchangeability.

%% file: Sections/03_online_cqr.tex
\section{Adaptive Calibration for CQR-Based Reserve Sizing}
\label{sec:cqr}

To maintain coverage as operating conditions evolve, we construct reserve
intervals using Adaptive Conformal Inference~(ACI) with CQR conformity
scores~\cite{gibbs2021adaptive}. The fitted conditional quantiles make the
intervals depend on operating conditions, while ACI uses observed misses to
update the miscoverage level that selects the conformal score quantile. The
resulting finite-horizon identity follows directly from the ACI recursion and
holds for arbitrary sequences of features, fitted models, and outcomes.

\subsection{ACI-CQR for Total Miscoverage}
\label{subsec:marginal_cqr}

Let $q^{\tau}(\bm{\phi})$ denote the conditional $\tau$-quantile of the net-load
deviation $\varepsilon$ given the operating features $\bm{\phi}$, and let
$\widehat q^{\tau}$ denote its fitted estimate. The corresponding lower and
upper quantile estimates for interval $(d,h)$ are
\begin{equation}
    \widehat q^{\mathrm{low}}_{d,h}
    :=\widehat q^{\alpha/2}_{d,h}(\bm{\phi}_{d,h}),
    \qquad
    \widehat q^{\mathrm{high}}_{d,h}
    :=\widehat q^{1-\alpha/2}_{d,h}(\bm{\phi}_{d,h}).
    \label{eq:cqr_base_quantiles}
\end{equation}
The corresponding CQR conformity score is
\begin{equation}
    S_{d,h}
    :=\max\left\{
        \widehat q^{\mathrm{low}}_{d,h}-\varepsilon_{d,h},
        \varepsilon_{d,h}-\widehat q^{\mathrm{high}}_{d,h}
    \right\}.
    \label{eq:online_aci_cqr_score}
\end{equation}
Adaptive calibration can regulate either one pooled sequence or separate
sequences associated with conditions known when reserves are scheduled. Let a
prespecified rule $G$ assign each interval $(d,h)$ to a calibration group
$g=G(d,h,\bm{\phi}_{d,h})\in\cG$. For any finite grouping rule $G$, the results
below apply separately to the sequence of observations in each group.

We consider two choices. Under marginal calibration, $G$ assigns every interval
to the same group, so all observations are pooled. Under hour-specific
calibration, $\cG=\cH$ and $G(d,h,\bm{\phi}_{d,h})=h$, so observations from each
hour form a separate calibration group. This hour-specific grouping follows the
scheduling structure considered here: the operating hour is known when reserves
are scheduled, and ERCOT procures ancillary services in the Day-Ahead Market by
service and hour of the Operating Day~\cite{ercotNodalProtocols2026}.
Calibrating by hour therefore evaluates reliability at the same temporal
resolution as reserve procurement. The theory also applies to any finite
grouping formed from information known when reserves are scheduled. Possible
alternatives include season and prespecified ranges of forecast load, net load,
renewable production, or ramp magnitude. Groups based on a continuous variable
require breakpoints and enough observations in every group. Finer groupings
provide more localized reliability. When a group receives fewer updates, its
finite-horizon residual term is correspondingly larger.

For $p\in(0,1]$, let $\widehat Q_{d,g}(p)$ denote the empirical $p$-quantile of
the past conformity scores available for group $g$ before the day-$d$ reserve
decisions. A score is included only after its corresponding outcome is
observed. For a fixed target total miscoverage $\alpha$, let
$\alpha_{d,g}\in\sR$ denote the time-varying ACI miscoverage state for group $g$
on day $d$. The same state is used to construct all applicable day-$d$
intervals. Its corresponding coverage level is $1-\alpha_{d,g}$, so ACI uses
the empirical score quantile
$\widehat Q_{d,g}(1-\alpha_{d,g})$. As in ACI, we extend the empirical quantile
outside $(0,1]$ as
\begin{equation}
    \widehat Q_{d,g}(p)
    =
    \begin{cases}
        -\infty,&p\leq0,\\
        +\infty,&p>1.
    \end{cases}
    \label{eq:aci_quantile_endpoints}
\end{equation}
The ACI-CQR interval is then
\begin{equation}
    C^{\mathrm{ACI}}_{d,h}
    =\left[
        \widehat q^{\mathrm{low}}_{d,h}-\widehat Q_{d,g}(1-\alpha_{d,g}),
        \widehat q^{\mathrm{high}}_{d,h}+\widehat Q_{d,g}(1-\alpha_{d,g})
    \right].
    \label{eq:online_aci_cqr_interval}
\end{equation}
After the outcomes for day $d$ are observed, define the ACI-CQR miss indicator
$M_{d,h}^{\mathrm{ACI}}:=
\indicator{\varepsilon_{d,h}\notin C^{\mathrm{ACI}}_{d,h}}$. ACI uses one
feedback value for each calibration group. Under marginal calibration, all
$H$ hours belong to one group, denoted by $g=\mathrm{all}$, and its feedback is
the daily miss rate
\begin{equation}
    M_{d,\mathrm{all}}^{\mathrm{ACI}}
    :=\frac{1}{H}\sum_{h\in\cH}M_{d,h}^{\mathrm{ACI}}.
    \label{eq:marginal_daily_miscoverage}
\end{equation}
Under hour-specific calibration, $g=h$ and
$M_{d,g}^{\mathrm{ACI}}:=M_{d,h}^{\mathrm{ACI}}$. Thus,
$M_{d,g}^{\mathrm{ACI}}\in[0,1]$ in either case.

Once the group's outcomes are observed, the state for the first eligible
schedule is updated using the group-level miss rate as
\begin{equation}
    \alpha_{d+2,g}
    =\alpha_{d+1,g}+\gamma_g\left(\alpha-M_{d,g}^{\mathrm{ACI}}\right),
    \label{eq:aci_cqr_update}
\end{equation}
where $\gamma_g>0$ is the ACI step size for group $g$ and determines how strongly the adaptive state responds to
the difference between the target and observed miss rates. This update may
temporarily place $\alpha_{d,g}$ outside $[0,1]$; the endpoint convention
in~\cref{eq:aci_quantile_endpoints} keeps every later interval defined. The new
scores are also added to the score collection used for day $d+2$ and later
decisions. The day-$(d+1)$ state and score collection have already been fixed
at the day-ahead cutoff, so day-$d$ outcomes cannot affect them.

\begin{theorem}[Finite-horizon total-miscoverage calibration]
    \label{theorem:online_aci_cqr}
    For either marginal or hour-specific calibration, fix a group $g$ and
    suppose $\alpha_{1,g},\alpha_{2,g}\in[0,1]$ and the score collection is nonempty whenever
    an interval is constructed. If the intervals are constructed using
    \cref{eq:aci_quantile_endpoints,eq:online_aci_cqr_interval} and the state
    is updated according to~\cref{eq:aci_cqr_update}, then, for every realized
    sequence and every $D\geq1$,
    \begin{equation}
        \frac{1}{D}\sum_{d=1}^{D}M_{d,g}^{\mathrm{ACI}}
        =\alpha+
        \frac{\alpha_{2,g}-\alpha_{D+2,g}}{\gamma_gD}.
        \label{eq:aci_cqr_identity}
    \end{equation}
    The states also satisfy
    $\alpha_{d,g}\in[-2\gamma_g(1-\alpha),1+2\gamma_g\alpha]$ for
    $d=1,\ldots,D+2$, and therefore
    \begin{equation}
        \left|
        \frac{1}{D}\sum_{d=1}^{D}M_{d,g}^{\mathrm{ACI}}-\alpha
        \right|
        \leq
        \frac{\max\{\alpha_{2,g},1-\alpha_{2,g}\}+2\gamma_g}
        {\gamma_gD}.
        \label{eq:aci_cqr_bound}
    \end{equation}
\end{theorem}

\Cref{theorem:online_aci_cqr} bounds the difference between the realized
time-average total miscoverage and the target $\alpha$ within each calibration
group. This difference is at most a term that
decreases as $1/D$. The guarantee holds for any realized sequence, so temporal
dependence and changes in net-load deviations do not invalidate it. Under
marginal calibration, the left-hand side of~\cref{eq:aci_cqr_identity} is the
pooled miscoverage over all $DH$ intervals. Under hour-specific calibration,
the theorem controls miscoverage for each hour $h$ across the $D$ days.
This is a time-average calibration guarantee, not a per-day or conditional
coverage guarantee. For another finite grouping, the same result follows by
indexing the feedback updates within each group and replacing $D$ with the
number of updates received by that group.

\subsection{Directional ACI-CQR}
\label{subsec:directional_aci_cqr}

ACI-CQR controls total interval miscoverage but does not determine how misses
are divided between the two tails. Consequently, the total miscoverage target
may be met even when misses are concentrated in one direction. To calibrate the
two directions separately, we extend ACI-CQR by applying separate ACI updates
to the lower- and upper-tail conformity scores. We refer to this construction
as Directional ACI-CQR:
\begin{equation}
    S^{\mathrm{long}}_{d,h}
    :=\widehat q^{\mathrm{low}}_{d,h}-\varepsilon_{d,h},
    \qquad
    S^{\mathrm{short}}_{d,h}
    :=\varepsilon_{d,h}-\widehat q^{\mathrm{high}}_{d,h}.
    \label{eq:directional_aci_cqr_scores}
\end{equation}
We maintain separate score collections and adaptive miscoverage states
$\alpha^{\mathrm{long}}_{d,g}$ and $\alpha^{\mathrm{short}}_{d,g}$. The
resulting interval is
\begin{equation}
    \begin{aligned}
        C^{\mathrm{D\text{-}ACI}}_{d,h}
        &:=\left[C^{\mathrm{low}}_{d,h},C^{\mathrm{high}}_{d,h}\right],\\
        C^{\mathrm{low}}_{d,h}
        &:=\widehat q^{\mathrm{low}}_{d,h}
          -\widehat Q^{\mathrm{long}}_{d,g}
            (1-\alpha^{\mathrm{long}}_{d,g}),\\
        C^{\mathrm{high}}_{d,h}
        &:=\widehat q^{\mathrm{high}}_{d,h}
          +\widehat Q^{\mathrm{short}}_{d,g}
            (1-\alpha^{\mathrm{short}}_{d,g}).
    \end{aligned}
    \label{eq:directional_aci_cqr_interval}
\end{equation}
Let $M^{\mathrm{long}}_{d,h}$ indicate an observation below the lower endpoint
of \cref{eq:directional_aci_cqr_interval}; $M^{\mathrm{short}}_{d,h}$ indicates
an observation above its upper endpoint. Their marginal daily or hour-specific
group losses are formed as in \cref{eq:marginal_daily_miscoverage}. Let
$\alpha^{\mathrm{short}},\alpha^{\mathrm{long}}\in(0,1)$ be the selected
miscoverage targets. They may be allocated asymmetrically. To
target total miscoverage $\alpha$, they can be selected so that
$\alpha^{\mathrm{short}}+\alpha^{\mathrm{long}}=\alpha$; the common symmetric
convention used in \cref{eq:cqr_base_quantiles} sets both targets to
$\alpha/2$. The directional states are updated as
\begin{equation}
    \alpha^{j}_{d+2,g}
    =\alpha^{j}_{d+1,g}+\gamma^{j}_{g}\left(\alpha^{j}-M^{j}_{d,g}\right),
    \qquad j\in\{\mathrm{short},\mathrm{long}\}.
    \label{eq:directional_aci_update}
\end{equation}
where $\gamma^{j}_{g}>0$ is the step size for exposure direction $j$ in group
$g$.

\begin{theorem}[Finite-horizon directional-miscoverage calibration]
    \label{theorem:directional_aci_cqr}
    Fix a calibration group $g$. Suppose the initial directional states satisfy
    $\alpha^{j}_{1,g},\alpha^{j}_{2,g}\in[0,1]$ for
    $j\in\{\mathrm{short},\mathrm{long}\}$ and the corresponding score
    collections are nonempty whenever an interval is constructed. If the
    intervals are constructed using
    \cref{eq:aci_quantile_endpoints,eq:directional_aci_cqr_interval} and the
    states are updated according to~\cref{eq:directional_aci_update}, then, for
    every realized sequence and every $D\geq1$,
    \begin{equation}
        \frac{1}{D}\sum_{d=1}^{D}M^{j}_{d,g}
        =\alpha^{j}
        +\frac{\alpha^{j}_{2,g}-\alpha^{j}_{D+2,g}}{\gamma^{j}_{g}D},
        \qquad j\in\{\mathrm{short},\mathrm{long}\},
        \label{eq:directional_aci_identity}
    \end{equation}
    Each directional state remains bounded as
    $\alpha^{j}_{d,g}\in
    [-2\gamma^{j}_{g}(1-\alpha^j),1+2\gamma^{j}_{g}\alpha^j]$ for
    $d=1,\ldots,D+2$. Consequently,
    \begin{equation}
        \left|
        \frac{1}{D}\sum_{d=1}^{D}M^{j}_{d,g}-\alpha^{j}
        \right|
        \leq
        \frac{\max\{\alpha^{j}_{2,g},1-\alpha^{j}_{2,g}\}
        +2\gamma^{j}_{g}}
        {\gamma^{j}_{g}D}.
        \label{eq:directional_aci_bound}
    \end{equation}
    If $M^{\mathrm{D\text{-}ACI}}_{d,g}$ denotes total miscoverage
    of~\cref{eq:directional_aci_cqr_interval}, then
    \begin{equation}
        \frac{1}{D}\sum_{d=1}^{D}M^{\mathrm{D\text{-}ACI}}_{d,g}
        \leq
        \frac{1}{D}\sum_{d=1}^{D}
        \left(M^{\mathrm{long}}_{d,g}+M^{\mathrm{short}}_{d,g}\right).
        \label{eq:directional_union_bound}
    \end{equation}
    Thus, the asymptotic upper limit of total time-average miscoverage is at
    most $\alpha^{\mathrm{long}}+\alpha^{\mathrm{short}}$ and, when the targets
    sum to $\alpha$, at most $\alpha$. The inequality
    in~\cref{eq:directional_union_bound} is an equality whenever the calibrated
    intervals are nonempty.
\end{theorem}

\Cref{theorem:directional_aci_cqr} proves separate finite-horizon bounds on
short- and long-exposure miscoverage within each calibration group.
For each direction, the difference between the realized average miss rate and
its selected target is bounded by a term that decreases as $1/D$, even when the
sequence of operating conditions and deviations changes over time. An
acceptable total coverage rate can still conceal persistent misses in one
direction; the separate guarantees expose and limit this imbalance within each
group. The short- and long-exposure miscoverage targets may differ to reflect
their different operational consequences. Setting the two targets equal does
not imply symmetric upward and downward reserves: the fitted quantile bounds
and score distributions may have different scales, so the reserve corrections
can differ.

\subsection{Mapping Intervals to Reserves}
\label{subsec:cqr_interval_to_reserves}

The intervals constructed above are expressed in terms of signed net-load
deviations. Let $C_{d,h}$ denote either $C^{\mathrm{ACI}}_{d,h}$ or
$C^{\mathrm{D\text{-}ACI}}_{d,h}$. For the selected construction, denote its
lower and upper endpoints by $C^{\mathrm{low}}_{d,h}$ and
$C^{\mathrm{high}}_{d,h}$, so that
$C_{d,h}=[C^{\mathrm{low}}_{d,h},C^{\mathrm{high}}_{d,h}]$. The corresponding
reserve quantities are
\begin{equation}
    r^{\mathrm{up}}_{d,h}
    :=\max\{0,C^{\mathrm{high}}_{d,h}\},
    \qquad
    r^{\mathrm{down}}_{d,h}
    :=\max\{0,-C^{\mathrm{low}}_{d,h}\}.
    \label{eq:cqr_reserves}
\end{equation}
The resulting operational interval contains the conformal interval,
$C_{d,h}\subseteq
[-r^{\mathrm{down}}_{d,h},r^{\mathrm{up}}_{d,h}]$. Consequently,
\begin{equation}
    \indicator{\varepsilon_{d,h}\notin
        [-r^{\mathrm{down}}_{d,h},r^{\mathrm{up}}_{d,h}]}
    \leq
    \indicator{\varepsilon_{d,h}\notin C_{d,h}}.
    \label{eq:reserve_coverage_transfer}
\end{equation}
Equation~\eqref{eq:reserve_coverage_transfer} transfers the upper bound on
time-average miscoverage from the conformal interval to the operational reserve
interval. When the mapping extends the interval to include zero, the operational
interval retains this guarantee but can have miscoverage below the target (better performance). If
the conformal interval already contains zero, \cref{eq:cqr_reserves} leaves it
unchanged, and the two intervals have the same miscoverage.

%% file: Sections/04_online_crc.tex
\section{Adaptive Calibration for CRC-Based Reserve Sizing}
\label{sec:crc}

Coverage measures how often reserve is insufficient, whereas EUD also reflects
how much deviation remains uncovered. We use Rolling Risk Control~(Rolling RC)
\cite{feldman2023rolling} to adapt upward and downward reserves so that their
directional EUD remains within operator-specified budgets. A quantile model,
tuned separately for each direction, provides the baseline reserve. Rolling RC
then adjusts this baseline through a calibration parameter updated from the
observed uncovered-deviation loss.
In the day-ahead implementation studied here, the calibration parameter is
updated once after each completed operating day and affects the first reserve
schedule formed after that feedback becomes available.

\subsection{Reserve Construction}
\label{subsec:eud}

We apply Rolling RC separately to the short and long exposure directions
introduced in~\cref{subsec:problem}. To express both directions using common
notation, define the directional deviation magnitudes
\begin{equation}
    Z^{\mathrm{short}}_{d,h}:=\max\{\varepsilon_{d,h},0\},
    \qquad
    Z^{\mathrm{long}}_{d,h}:=\max\{-\varepsilon_{d,h},0\},
    \label{eq:directional_targets}
\end{equation}
and denote the corresponding reserves by
$r^{\mathrm{short}}_{d,h}:=r^{\mathrm{up}}_{d,h}$ and
$r^{\mathrm{long}}_{d,h}:=r^{\mathrm{down}}_{d,h}$.
With $j\in\{\mathrm{short},\mathrm{long}\}$, the realized uncovered deviation
in either direction is
$(Z^j_{d,h}-r^j_{d,h})^+$, consistent with
\cref{subsec:eud_metric}.

Rolling RC operates on a base quantile model that provides a baseline reserve
for each direction:
\begin{equation}
    \widehat r^j_{d,h}
    :=\max\!\left\{0,
        \widehat q^{j,\tau_j}_{d,h}(\bm{\phi}_{d,h})\right\},
    \label{eq:base_directional_reserve}
\end{equation}
where $\widehat q^{j,\tau_j}_{d,h}(\bm{\phi}_{d,h})$ predicts the
$\tau_j$-quantile of the directional deviation magnitude under the current
operating conditions. The level $\tau_j$ determines the baseline reserve and
may be tuned on validation data so that its mean uncovered deviation is close
to the EUD budget $\rho^j$. This follows
the validation-based parameter-selection principle used for monotone losses in
conformal risk control~\cite{angelopoulos2024conformal}. The choice of
$\tau_j$ affects the baseline reserve and the subsequent adjustment required
from Rolling RC, but not the form of its guarantee.

Rolling RC assigns a separate calibration parameter to each direction and
group. Using the grouping rule $G$ introduced in~\cref{subsec:marginal_cqr},
interval $(d,h)$ belongs to group $g=G(d,h,\bm{\phi}_{d,h})$, and
$\theta^j_{d,g}\in\sR$ denotes the parameter used to determine the day-$d$ reserve
in direction $j$. The adjusted reserve is
\begin{equation}
    r^j_{d,h}(\theta^j_{d,g})
    :=\max\left\{0,
        \widehat r^j_{d,h}+\varphi_j(\theta^j_{d,g})
    \right\},
    \label{eq:online_reserve_decision}
\end{equation}
where $\varphi_j$ is a nondecreasing stretching function that maps the calibration
parameter to a reserve adjustment; a larger $\theta^j_{d,g}$ cannot reduce the
reserve. No stretching corresponds to
$\varphi_j(\theta)=\theta$; exponential and error-adaptive functions may also be
used~\cite{feldman2023rolling}.

\subsection{Sequential Update and Finite-Horizon Guarantee}
\label{subsec:rolling_rc}

Rolling RC updates the calibration parameter $\theta^j_{d,g}$ once after all
deviations for operating day $d$ are observed, and its guarantee requires a
bounded loss~\cite{feldman2023rolling}. For
$j\in\{\mathrm{short},\mathrm{long}\}$, let $B^j_g$ be the fixed loss cap for
direction $j$ and group $g$, selected before the outcomes are known and
satisfying $\rho^j<B^j_g<\infty$. The loss used in the update is the capped
uncovered deviation
\begin{equation}
    L^j_{d,h}(\theta)
    :=\min\left\{
        \left(Z^j_{d,h}-r^j_{d,h}(\theta)\right)^+,
        B^j_g
    \right\}.
    \label{eq:bounded_eud_loss}
\end{equation}
The cap is applied only when evaluating the loss; it does not limit the reserve
in~\cref{eq:online_reserve_decision} or the physical uncovered deviation.

After the outcomes for day $d$ are observed, the interval losses are combined
according to the selected grouping. Under marginal calibration, all hours
belong to $g=\mathrm{all}$ and the loss is the daily mean
\begin{equation}
    L^j_{d,\mathrm{all}}(\theta)
    :=\frac{1}{H}\sum_{h\in\cH}L^j_{d,h}(\theta).
    \label{eq:daily_marginal_eud}
\end{equation}
Under hour-specific calibration, each hour is treated separately, so $g=h$ and
$L^j_{d,g}(\theta):=L^j_{d,h}(\theta)$. In either case,
$0\leq L^j_{d,g}(\theta)\leq B^j_g$. The objective for each direction
and group is to keep the time-average loss within the EUD budget:
\begin{equation}
    \limsup_{D\rightarrow\infty}
    \frac{1}{D}\sum_{d=1}^{D}
        L^j_{d,g}(\theta^j_{d,g})
    \leq\rho^j.
    \label{eq:risk_sizing}
\end{equation}
Once $L^j_{d,g}$ is available, the calibration parameter is updated as
\begin{equation}
    \begin{aligned}
        \theta^j_{d+2,g}
        &=\theta^j_{d+1,g}
        +\gamma^j_g\left(
            L^j_{d,g}(\theta^j_{d,g})-\rho^j
        \right),\\
        &\hspace{26mm}j\in\{\mathrm{short},\mathrm{long}\},
    \end{aligned}
    \label{eq:online_crc}
\end{equation}
where $\gamma^j_g>0$ is the step size. A loss above the budget increases the
parameter and the reserve adjustment for subsequent days. A loss below the
budget decreases the parameter, allowing a smaller subsequent adjustment. Thus,
under the stated day-ahead cutoff, outcomes from day $d$ first affect the
reserves issued for day $d+2$; the day-$(d+1)$ decision is already fixed. The
condition $B^j_g>\rho^j$ ensures that the capped loss can exceed the budget;
otherwise, the update could never increase the calibration parameter.

\begin{theorem}[Finite-horizon Rolling RC]
    \label{theorem:rolling_rc}
    Fix a group $g$ and a direction
    $j\in\{\mathrm{short},\mathrm{long}\}$. Suppose there is a finite value
    $\theta^j_{\mathrm{safe},g}$ such that
    \begin{equation}
        L^j_{d,g}(\theta)\leq\rho^j
        \quad\text{for all }d
        \text{ and }\theta\geq\theta^j_{\mathrm{safe},g}.
        \label{eq:rrc_safe_state}
    \end{equation}
    Then, for every realized sequence and every $D\geq1$, the update
    in~\cref{eq:online_crc} satisfies
    \begin{equation}
        \frac{1}{D}\sum_{d=1}^{D}
        L^j_{d,g}(\theta^j_{d,g})
        =\rho^j+
        \frac{\theta^j_{D+2,g}-\theta^j_{2,g}}
        {\gamma^j_gD}.
        \label{eq:rrc_identity}
    \end{equation}
    For every $d\geq1$, the state also satisfies
    \begin{equation}
        \begin{aligned}
        \theta^j_{d,g}&\leq\Theta^j_g,\\
        \Theta^j_g&:=\max\left\{
        \begin{array}{l}
            \theta^j_{1,g},\quad
            \theta^j_{2,g}+\gamma^j_g(B^j_g-\rho^j),\\
            \theta^j_{\mathrm{safe},g}
            +2\gamma^j_g(B^j_g-\rho^j)
        \end{array}
        \right\}.
        \end{aligned}
        \label{eq:rrc_state_bound}
    \end{equation}
    and therefore
    \begin{equation}
        \frac{1}{D}\sum_{d=1}^{D}
        L^j_{d,g}(\theta^j_{d,g})
        \leq\rho^j+
        \frac{\Theta^j_g-\theta^j_{2,g}}
        {\gamma^j_gD}.
        \label{eq:rrc_finite_bound}
    \end{equation}
    Consequently, the time-average objective in~\cref{eq:risk_sizing} holds.
\end{theorem}

Under the uniform feasibility condition in~\cref{eq:rrc_safe_state},
\Cref{theorem:rolling_rc} proves a finite-horizon bound on the time-average
capped uncovered-deviation loss for each direction and group by the selected
budget plus a residual that decreases as $1/D$. The bound applies to every realized
sequence satisfying this condition, without requiring stationarity or
exchangeability. The loss cap ensures bounded feedback, but does not by itself
establish~\cref{eq:rrc_safe_state} or provide a corresponding bound on uncapped
EUD.
The day-$d+2$ protocol introduces a one-operating-day feedback lag. Relative to
immediate feedback, this lag doubles the worst-case overshoot term
in~\cref{eq:rrc_state_bound} but does not change the $O(1/D)$ rate or the
asymptotic guarantee.

Under marginal calibration, the average
in~\cref{eq:rrc_identity} is the capped EUD loss pooled over all $DH$
intervals. Under hour-specific calibration, the guarantee holds separately for
each hour across the $D$ days.
The one-sided bound reflects the objective of keeping EUD at or below its
budget. It applies to any stretching function for which a sufficiently large
calibration parameter satisfies~\cref{eq:rrc_safe_state}. The system operator
should select $B^j_g$ to reflect credible uncovered deviations, informed by
historical extremes and stress scenarios. A smaller cap limits the influence of
extreme observations, while a larger cap permits larger updates and increases
the finite-horizon residual. If reserve adjustments are capped, the maximum
adjustment must still allow~\cref{eq:rrc_safe_state} to hold. Otherwise, the
identity in~\cref{eq:rrc_identity} remains valid, but the finite-horizon upper
bound is not guaranteed.

%% file: Sections/05_experimental_design.tex
\section{Numerical Experiments}
\label{sec:experiments}

We evaluate the proposed adaptive reserve-calibration frameworks using coverage
and EUD as reliability criteria. Each framework is compared with uncalibrated
quantile reserves and a rolling-window calibration benchmark on a chronological
ERCOT sequence. We consider marginal and hour-specific calibration and assess
both target attainment and the required reserve.

\subsection{Data and Experimental Setup}
\label{subsec:data_experimental_setup}

\subsubsection{Dataset}

We use four years of hourly ERCOT data from 2022 through 2025. As shown in~\cref{fig:forecast_errors}, the
dispersion and tails of the net-load deviations vary across seasons and years.

We construct an hourly dataset from ERCOT's public reports~\cite{ercotdata}.
For each operating hour, forecast net load equals the value from ERCOT's
rolling seven-day system-load forecast minus the median (P50) values from its
short-term wind and solar forecasts. ERCOT's cutoff for submitting Day-Ahead
Market inputs for day $d$ is 10:00~h on day $d-1$. We therefore use the latest
forecasts published before that cutoff, so the day-ahead values contain only
information available at scheduling~\cite{ercotNodalProtocols2026}. Realized
net load equals official Native Load minus reported wind and solar generation,
and the hourly net-load deviation is realized net load minus forecast net load.
Behind-the-meter generation is not available as a separate series; its net
effect is reflected implicitly in aggregate system load when it reduces the
demand observed by ERCOT.

\subsubsection{Base Model and Features}

The model inputs are forecast load, wind generation, solar generation, derived
net load, hour of day, day of year, and day of week. Together, these features
describe the expected system conditions and calendar patterns known when the
reserves are scheduled. Because the feature set combines continuous forecasts
with discrete calendar variables, we estimate the conditional quantiles with
XGBoost, which can accommodate both directly~\cite{chen2016xgboost}. This avoids
the feature scaling and calendar encodings typically introduced for neural
networks.

The model choice is not intrinsic to the framework. Linear quantile regression,
quantile regression forests, neural networks, or any other learner capable of
estimating conditional quantiles may be used instead
\cite{koenker1978regression,meinshausen2006quantile,cannon2011quantile}. The
finite-horizon guarantees in
\cref{theorem:online_aci_cqr,theorem:directional_aci_cqr,theorem:rolling_rc}
are independent of this choice, so replacing XGBoost does not change the adaptive
calibration procedures.

\subsubsection{Calibration Groupings}

We compare marginal calibration, which pools all hours, with hour-specific
calibration, which applies the reliability target separately to each operating
hour. Other groupings may be selected according to the operator's reliability
needs and the data available within each group.

\subsubsection{Chronological Calibration Protocol}
\label{subsubsec:chronological_calibration}

The 2022--2023 observations form the historical period used for model training,
validation, and initialization. The 2024--2025 observations form one
uninterrupted chronological test sequence. The models used on the first test
day are fitted with all available historical observations and are subsequently
refitted each month using an expanding window that contains only observations
available by the applicable day-ahead cutoff.

For operating day $d+1$, ERCOT's day-ahead inputs are due at 10:00 on day $d$
\cite{ercotNodalProtocols2026}, before the complete outcomes of day $d$ are
available. Accordingly, outcomes, scores, and model-training observations from
day $d$ are first used for the
day-$d+2$ schedule. The interval between the end of day $d$ and the next
day-ahead cutoff provides time for the data to be validated and processed. This
one-operating-day feedback lag is a realistic deployment protocol for the
stated cutoff and prevents information leakage. A market with a different
cutoff or processing delay can use the corresponding feedback lag without
changing the calibration framework. The adaptive states carry continuously
across 2024 and 2025, preserving the heteroscedasticity and distribution shifts
in the observed sequence.

%% file: Sections/06_results.tex
\subsection{Coverage Calibration Results}
\label{subsec:coverage_results}

We compare ACI-CQR and Directional ACI-CQR with uncalibrated quantile reserves
(``Raw'') and rolling-window CQR (``Window CQR'') over the chronological
2024--2025 test. All reported rates and mean reserves are calculated over the
complete test period. All methods target 90\% coverage. ``Raw'' uses the fitted
conditional 5\textsuperscript{th} and 95\textsuperscript{th} quantiles, with
models refitted monthly on an expanding window. ``Window CQR'' recalculates its
10\% conformity correction from the latest 100 complete days; its standard
guarantee relies on exchangeability~\cite{romano2019conformalized}, which is not
supported by the distribution changes in \cref{fig:forecast_errors}.

The adaptive methods retain the same 100-day score history. ACI-CQR maintains
one 10\% state per group, whereas Directional ACI-CQR maintains downward and
upward states with 5\% targets. Marginal calibration pools all intervals in
the window, while each hour-specific group receives one score per day.
Following~\cite{gibbs2021adaptive}, all step sizes equal 0.005 and are fixed
before testing; the value balances responsiveness with variation in the
adaptive level. One hundred historical days initialize the score collections
and the next 100 are processed sequentially to initialize the states and
advance the score histories.
Subsequent scores and state updates first enter the day-$d+2$ schedule according
to~\cref{eq:aci_cqr_update,eq:directional_aci_update}.

\Cref{fig:cqr_directional} reports the pooled results. ``Raw'' has downward and
upward miss rates of 4.95\% and 9.27\%, yielding only 85.79\% overall coverage.
``Window CQR'' and ACI-CQR raise coverage to 90.26\%--90.45\%, but their
directional rates remain uneven at about 3.3\% downward and 6.4\% upward.
Directional ACI-CQR instead obtains 4.60\% and 4.85\% under marginal
calibration and 4.67\% and 4.96\% under hour-specific calibration. It is the
only method that meets the 5\% upward target under both groupings, showing that
total calibration alone does not control the allocation of misses. In
particular, an apparently satisfactory 90\% total coverage rate can combine a
conservative downward reserve with an undercovered upward reserve. The
directional update makes that allocation explicit and regulates the two
physical exposure directions separately.

``Raw'' uses 7,535~MW of mean total reserve. Under marginal calibration,
``Window CQR,'' ACI-CQR, and Directional ACI-CQR use 8,520, 8,542, and
8,571~MW, respectively. The latter total comprises 3,117~MW downward and
5,454~MW upward. Directional ACI-CQR therefore corrects the upward undercoverage
mainly by reallocating reserve, not by materially increasing the total relative
to the other calibrated methods. When undercoverage is concentrated in one
direction, reallocating reserve can improve reliability with little change in
the total requirement.

\begin{figure}[t]
    \centering
    \includegraphics[width=\linewidth]{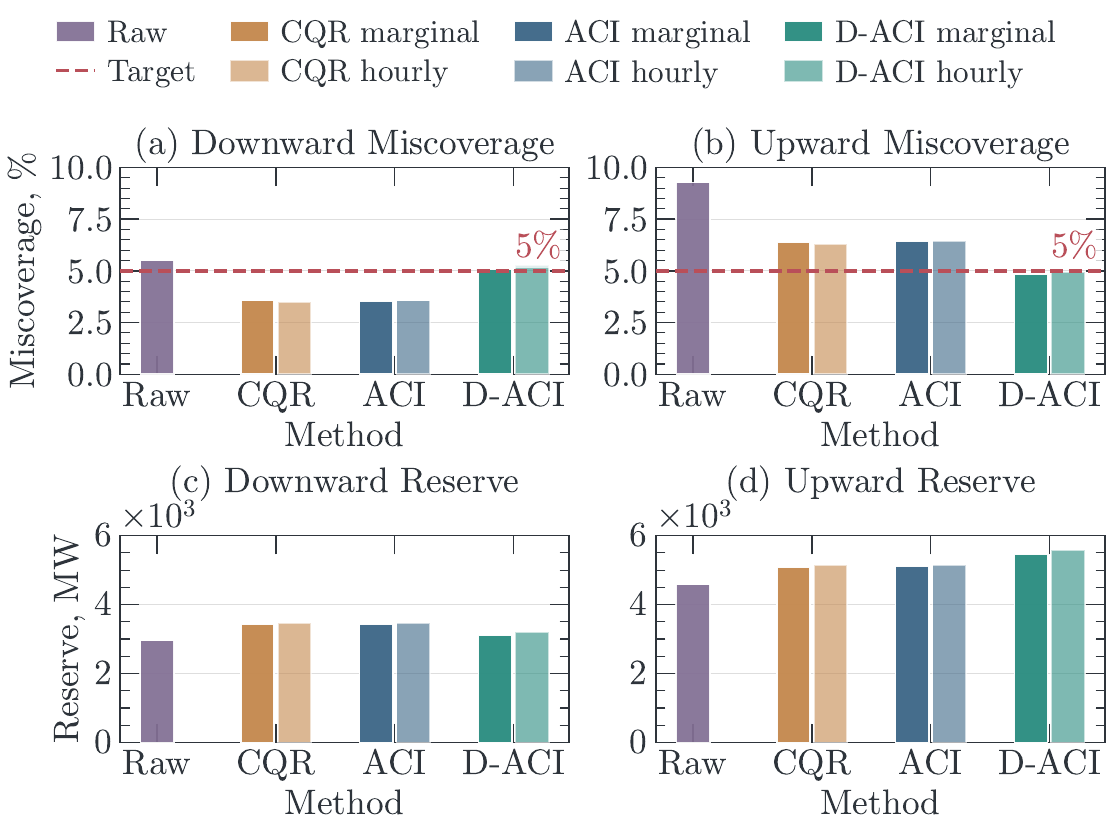}
    \caption{ERCOT pooled 2024--2025 coverage-calibration results: (a) downward
    miscoverage; (b) upward miscoverage; (c) mean downward reserve; and
    (d) mean upward reserve. CQR, ACI, and D-ACI denote ``Window CQR,'' ACI-CQR,
    and Directional ACI-CQR, respectively. Dashed lines mark the 5\% directional
    targets.}
    \Description{Four grouped-bar panels compare ``Raw,'' ``Window CQR,'' ACI-CQR, and
    Directional ACI-CQR under marginal and hourly calibration. The upper panels
    show downward- and upward-tail miscoverage, and the lower panels show mean
    downward and upward reserves.}
    \label{fig:cqr_directional}
\end{figure}

\Cref{fig:cqr_hourly} exposes the variation hidden by pooled rates. Under
marginal Directional ACI-CQR, hourly directional miss rates span
2.48\%--8.67\%; hour-specific calibration narrows the range to
3.71\%--5.64\%. Its mean total reserve is 8,771~MW, only 200~MW (2.33\%) above
marginal Directional ACI-CQR. Thus, a small reserve premium replaces hourly
overcoverage and undercoverage with substantially more consistent reliability.

\begin{figure}[t]
    \centering
    \includegraphics[width=\linewidth]{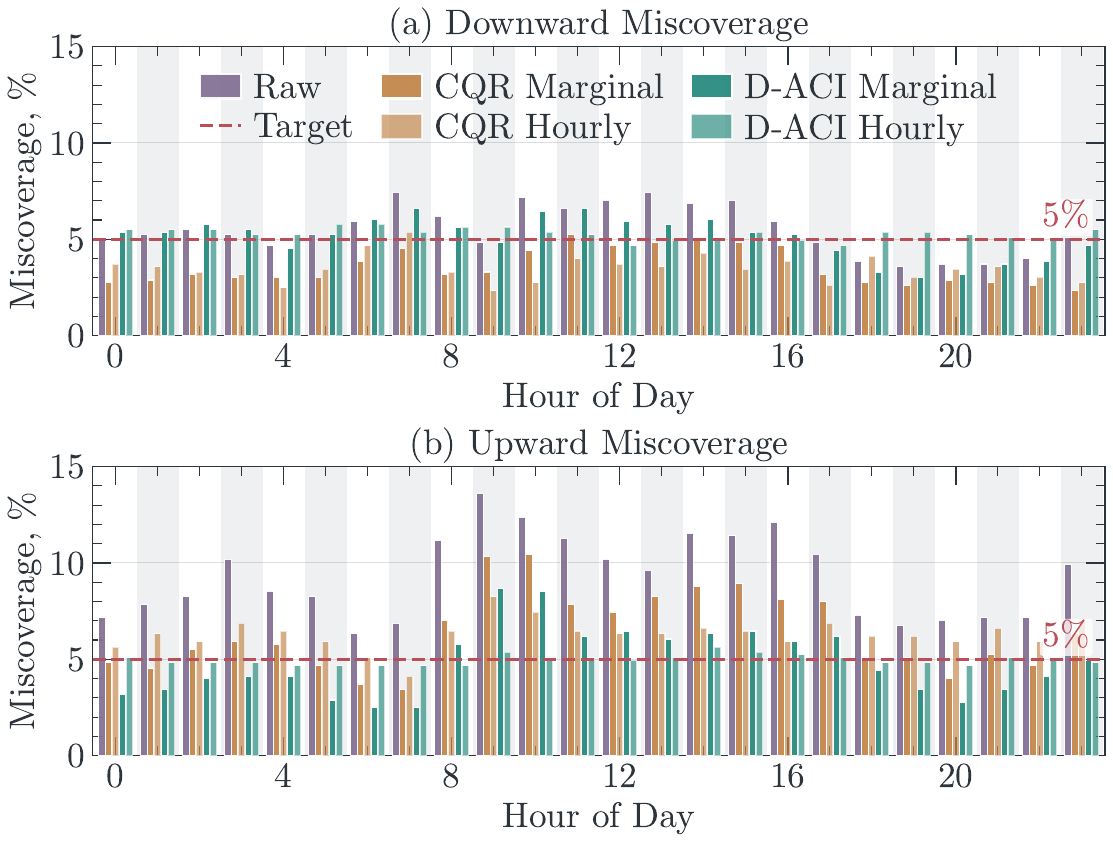}
    \caption{ERCOT directional miscoverage by operating hour: (a) downward and
    (b) upward. CQR and D-ACI denote ``Window CQR'' and Directional ACI-CQR,
    respectively. Hourly calibration uses 24 hour-specific states; dashed lines
    mark the 5\% directional targets.}
    \Description{Two grouped-bar panels compare ``Raw,'' ``Window CQR,'' and Directional
    ACI-CQR miscoverage under marginal and hourly calibration across the 24 ERCOT
    operating hours.}
    \label{fig:cqr_hourly}
\end{figure}

\subsection{EUD Calibration Results}
\label{subsec:eud_results}

We compare Rolling RC with two baselines using the same test period, groupings,
and model-refitting protocol as the coverage study.

The ``Raw'' baseline uses directional quantile levels selected from historical
EUD budgets, unlike the fixed 5\textsuperscript{th} and 95\textsuperscript{th}
quantiles in the coverage study. Reserve savings are therefore evaluated
against the baseline within each study; the two ``Raw'' reserve totals are not
directly comparable.
``Window CRC'' recalculates its adjustment daily from the latest 100 complete
days~\cite{angelopoulos2024conformal}. The day-$d+2$ window ends with day $d$,
and its standard guarantee relies on exchangeability.

EUD budgets depend on system size, deviation magnitude, security criteria,
corrective actions, deliverability, cost, and the consequences of each exposure
direction~\cite{ercotASMethodology2026,caisoBPM2026,nyisoASManual2026}. An
operator should therefore select the budgets with the corrective actions and
consequences of short and long exposure in mind; they need not be equal.
Unlike a coverage level, an EUD budget has physical units and is consequently
system dependent. The values used here are experimental reliability targets,
not universal operating recommendations.

We use marginal Directional ACI-CQR at 90\% coverage on the 2022--2023
validation data as the reference, setting the mean uncovered-deviation budgets
to
$\rho^{\mathrm{short}}=55$~MW and $\rho^{\mathrm{long}}=48$~MW.
These budgets describe the expected residual after reserve is applied, not the
magnitude of the underlying deviations.
The corresponding baseline quantile levels are
$\tau_{\mathrm{short}}=0.9735$ and $\tau_{\mathrm{long}}=0.9858$. All levels
and budgets are fixed before testing and shared by both groupings.

We tune $\gamma\in\{0.1,0.25,0.5,1,2\}$ on the latest 100 complete historical
days, selecting the lowest-reserve candidate whose EUD is no more than 2~MW
above budget, separately by hour when applicable. If none qualifies, we
minimize the maximum exceedance. Test observations do not influence this
selection. This gives $\gamma^{\mathrm{short}}_g=0.1$ and
$\gamma^{\mathrm{long}}_g=0.5$ for marginal calibration, and
$\gamma^{\mathrm{short}}_g=0.5$ and $\gamma^{\mathrm{long}}_g=1$ for
hour-specific calibration. We use $\varphi_j(\theta)=\theta$ and set $B^j_g$
to the largest corresponding 2022--2023 deviation. All choices precede testing.

In \cref{fig:crc_directional}, ``Raw'' test EUD is 18.93~MW downward and
73.31~MW upward, below and above the respective budgets. ``Window CRC'' exceeds
the marginal upward budget by 5.14~MW and the hour-specific downward budget by
4.99~MW. Rolling RC reduces the marginal upward exceedance to 1.58~MW; its
hour-specific EUD values, 47.66~MW downward and 53.21~MW upward, meet both
pooled budgets. Under the observed distribution changes, historical tuning
alone leaves one direction substantially underprotected even when the other
remains conservative.

``Raw'' uses 9,536~MW of mean total reserve. Marginal Rolling RC uses 9,027~MW,
5.3\% less than ``Raw'' and 1.2\% less than marginal ``Window CRC,'' while
shifting reserve upward to bring both EUD values closer to budget. Its
downward and upward components are 3,443 and 5,584~MW, compared with 4,299 and
5,237~MW for ``Raw.''

\begin{figure}[t]
    \centering
    \includegraphics[width=\linewidth]{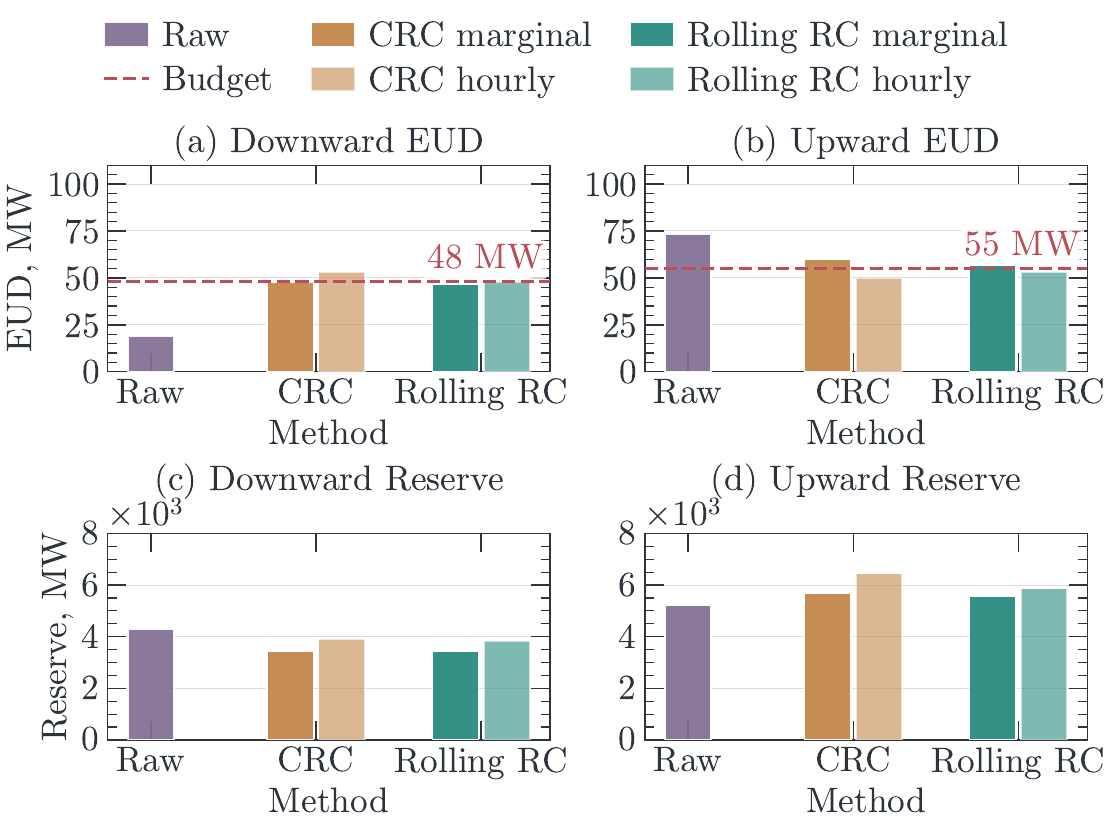}
    \caption{ERCOT pooled 2024--2025 EUD-calibration results: (a) downward EUD;
    (b) upward EUD; (c) mean downward reserve; and (d) mean upward reserve. CRC
    denotes ``Window CRC.'' Dashed lines mark the directional EUD budgets.}
    \Description{Four grouped-bar panels compare ``Raw'' risk, ``Window CRC,'' and
    Rolling Risk Control under marginal and hour-specific calibration. The upper panels
    show downward and upward EUD, and the lower panels show mean downward and
    upward reserves.}
    \label{fig:crc_directional}
\end{figure}

\Cref{fig:crc_hourly} shows that satisfactory pooled EUD can coexist with
substantial budget exceedances at individual operating hours. Hour-specific
Rolling RC reduces the largest hourly exceedance from 69.6 to 6.3~MW, with
7.4\% more mean total reserve than marginal Rolling RC. This is the observed
reserve premium for enforcing the EUD budgets separately at each hour. Within
the same hour-specific grouping, Rolling RC requires 6.4\% less reserve than
``Window CRC.'' These comparisons separate two decisions: the temporal
resolution at which reliability is controlled and the calibration method used
at that resolution. The first determines the reserve premium for localizing
the target; the second determines how efficiently that target is tracked.

\begin{figure}[t]
    \centering
    \includegraphics[width=\linewidth]{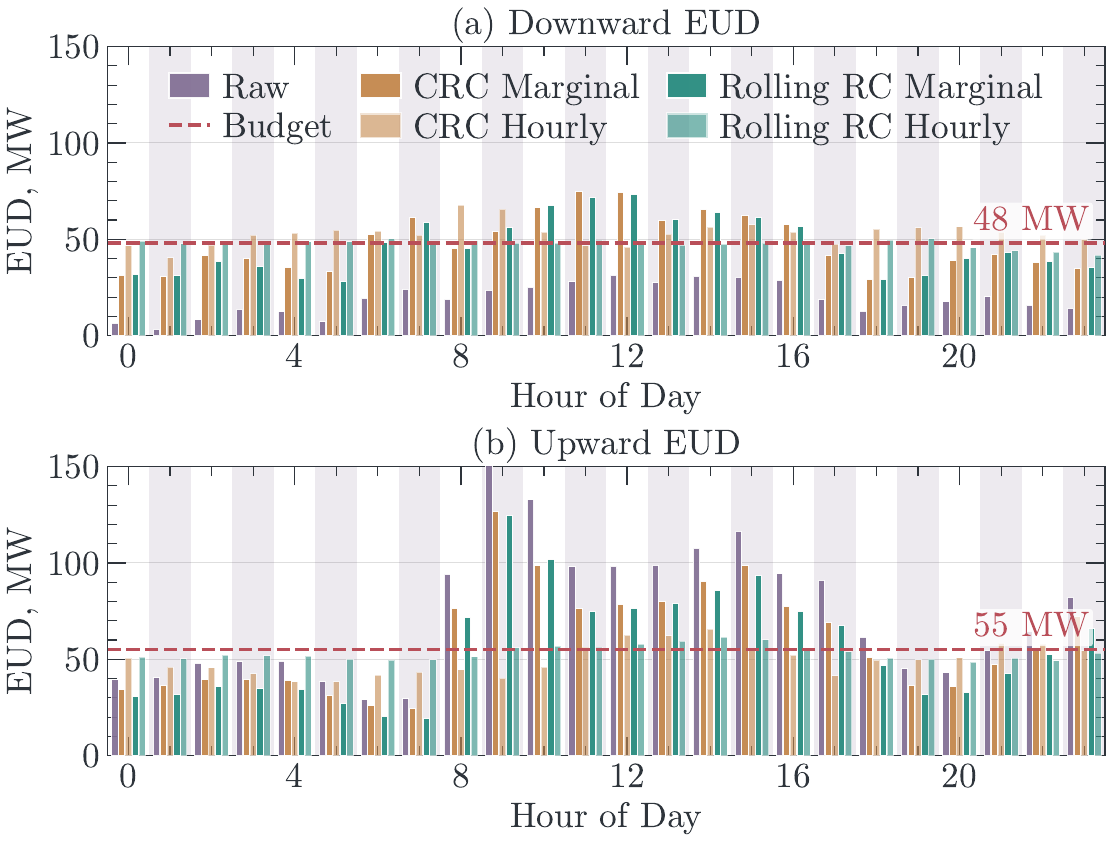}
    \caption{ERCOT directional EUD by operating hour: (a) downward and
    (b) upward. CRC denotes ``Window CRC.'' Marginal and hour-specific calibration are
    compared; dashed lines mark the directional EUD budgets.}
    \Description{Two grouped-bar panels compare ``Raw,'' ``Window CRC,'' and Rolling
    Risk Control EUD under marginal and hour-specific calibration across the 24 ERCOT
    operating hours.}
    \label{fig:crc_hourly}
\end{figure}

\subsection{Operational Scope}
\label{subsec:deployment_scope}

The adaptive methods regulate selected directional miss rates or capped EUD
losses; they do not establish economic optimality or security sufficiency.
Reserve volume is not production cost, and this system-level study does not
model nodal deliverability, outages, sub-hourly ramps, bids, or market
responses. Deployment therefore requires operator-selected targets and bounds,
security-constrained testing, fixed decision and feedback timing, and
prospective monitoring.

%% file: Sections/07_conclusion.tex
\section{Conclusion}
\label{sec:conclusion}

This work presents a sequential calibration framework for sizing flexibility
reserves under changing net-load deviation distributions. ACI-CQR calibrates total
miscoverage, Directional ACI-CQR applies separate short- and long-exposure
updates, and Rolling RC calibrates reserves against directional EUD budgets.
We prove finite-horizon bounds on realized time-average miscoverage and capped
EUD loss without assuming stationarity or exchangeability. The day-$d+2$
feedback protocol also shows how the methods can respect a day-ahead
scheduling cutoff without using outcomes that were unavailable when the
schedule was formed.

The chronological ERCOT evaluation shows that reserves estimated from
historical data can depart from their coverage and EUD targets during
deployment, even with monthly model retraining. Directional ACI-CQR achieves
approximately 90\% reserve coverage with upward and downward miscoverage
near 5\%. Hour-specific calibration improves reliability across hours with only
2.3\% more total reserve than marginal calibration. Marginal Rolling RC
brings pooled EUD close to both budgets while using 5.3\% less total reserve
than the ``Raw'' method. Hour-specific Rolling RC reduces the largest hourly
EUD-budget exceedance from 69.6 to 6.3~MW, with 7.4\% more reserve than
marginal Rolling RC and 6.4\% less than hour-specific ``Window CRC.''

The approach supplies deterministic upward and downward reserve requirements,
so it can enter an existing scheduling model through updated reserve inputs.
It is not tied to XGBoost, hour-of-day grouping, or the day-ahead-to-real-time
horizon; other learners, finite operational groups, and scheduling horizons
can be used when their information and feedback timing is defined. The method
controls the selected reliability metrics, not production cost or security
sufficiency. Deployment requires operator-selected targets and bounds, fixed
feedback timing, and security-constrained testing. Future work should evaluate
costs, nodal deliverability, outages, sub-hourly ramps, bids, and market
responses.

%% file: Appendices/proofs.tex
\section{Proofs}
\label{app:proofs}

\subsection{Total-Miscoverage Calibration}
\begin{proof}[Proof of \cref{theorem:online_aci_cqr}]
    Summing~\cref{eq:aci_cqr_update} over $d=1,\ldots,D$ gives
    \begin{equation*}
        \alpha_{D+2,g}-\alpha_{2,g}
        =\gamma_g\left(D\alpha-
        \sum_{d=1}^{D}M_{d,g}^{\mathrm{ACI}}\right),
    \end{equation*}
    which yields~\cref{eq:aci_cqr_identity} after rearranging.

    To bound the miscoverage deviation in~\cref{eq:aci_cqr_identity}, we first
    show that the adaptive state remains in
    $[-2\gamma_g(1-\alpha),1+2\gamma_g\alpha]$. Because
    $M_{d,g}^{\mathrm{ACI}}\in[0,1]$, an update increment lies in
    $[-\gamma_g(1-\alpha),\gamma_g\alpha]$. Starting from
    $\alpha_{2,g}\in[0,1]$, the first delayed update already lies inside the
    claimed interval. For $d\geq2$, if
    $\alpha_{d+1,g}>1+\gamma_g\alpha$, then the preceding increment bound
    implies $\alpha_{d,g}>1$. The endpoint convention therefore makes the
    day-$d$ interval empty, so $M_{d,g}^{\mathrm{ACI}}=1$ and the update to
    $\alpha_{d+2,g}$ decreases the state. Otherwise, the largest possible
    increment gives $\alpha_{d+2,g}\leq1+2\gamma_g\alpha$.

    The lower bound follows symmetrically. If
    $\alpha_{d+1,g}< -\gamma_g(1-\alpha)$, then the preceding increment bound
    implies $\alpha_{d,g}<0$. The day-$d$ interval is then $\sR$, so
    $M_{d,g}^{\mathrm{ACI}}=0$ and the update increases the state. Otherwise,
    the smallest possible increment gives
    $\alpha_{d+2,g}\geq-2\gamma_g(1-\alpha)$. Starting from
    $\alpha_{1,g},\alpha_{2,g}\in[0,1]$, induction establishes the stated
    state bound. Finally,
    \begin{equation*}
        |\alpha_{2,g}-\alpha_{D+2,g}|
        \leq \max\{\alpha_{2,g},1-\alpha_{2,g}\}+2\gamma_g.
    \end{equation*}
    Combining this inequality with~\cref{eq:aci_cqr_identity}
    gives~\cref{eq:aci_cqr_bound}.
\end{proof}

\subsection{Directional-Miscoverage Calibration}
\begin{proof}[Proof of \cref{theorem:directional_aci_cqr}]
    Fix $j\in\{\mathrm{short},\mathrm{long}\}$. Summing
    \cref{eq:directional_aci_update} over $d=1,\ldots,D$ and rearranging gives
    \cref{eq:directional_aci_identity}. The endpoint convention determines the
    directional miss whenever the state leaves $[0,1]$. The same two-step
    argument as in the total-miscoverage proof, with target $\alpha^j$ and
    step size $\gamma^j_g$, gives
    \begin{equation*}
        \alpha^j_{d,g}\in
        [-2\gamma^j_g(1-\alpha^j),1+2\gamma^j_g\alpha^j].
    \end{equation*}
    The first delayed update from $\alpha^j_{2,g}\in[0,1]$ lies inside this
    interval. Thereafter, a state above $1+\gamma^j_g\alpha^j$ implies that the state
    used for the corresponding day-$d$ decision exceeded one, forcing
    $M^j_{d,g}=1$; a state below $-\gamma^j_g(1-\alpha^j)$ implies that the
    decision state was negative, forcing $M^j_{d,g}=0$. The next delayed update
    therefore moves inward in either case. Starting from the two initial states
    in $[0,1]$ proves the bound by induction. Substituting it into
    \cref{eq:directional_aci_identity} gives
    \cref{eq:directional_aci_bound}.

    A total miss belongs to the union of the short- and long-exposure miss
    events, which gives \cref{eq:directional_union_bound}. When the calibrated
    interval is nonempty, these events are disjoint, and the inequality becomes
    an equality. Combining the directional bounds and taking the limit superior
    establishes the final claim.
\end{proof}

\subsection{Rolling Risk Control}
\begin{proof}[Proof of \cref{theorem:rolling_rc}]
    Summing~\cref{eq:online_crc} over $d=1,\ldots,D$ gives
    \begin{equation*}
        \theta^j_{D+2,g}-\theta^j_{2,g}
        =\gamma^j_g\sum_{d=1}^{D}
        \left(L^j_{d,g}(\theta^j_{d,g})-\rho^j\right).
    \end{equation*}
    Rearranging gives~\cref{eq:rrc_identity}. Let
    $K^j_g:=\gamma^j_g(B^j_g-\rho^j)$, the largest possible upward increment.
    The first delayed update satisfies
    $\theta^j_{3,g}\leq\theta^j_{2,g}+K^j_g$. For $d\geq2$, if
    $\theta^j_{d+1,g}\geq\theta^j_{\mathrm{safe},g}+K^j_g$, then the preceding
    increment bound implies
    $\theta^j_{d,g}\geq\theta^j_{\mathrm{safe},g}$. Hence,
    \cref{eq:rrc_safe_state,eq:online_crc} give
    $\theta^j_{d+2,g}\leq\theta^j_{d+1,g}$. If instead
    $\theta^j_{d+1,g}<\theta^j_{\mathrm{safe},g}+K^j_g$, then
    \begin{equation*}
        \theta^j_{d+2,g}
        \leq\theta^j_{d+1,g}+K^j_g
        <\theta^j_{\mathrm{safe},g}+2K^j_g.
    \end{equation*}
    Starting from $\theta^j_{1,g}$, $\theta^j_{2,g}$, and the stated bound on
    $\theta^j_{3,g}$, these two cases establish
    \cref{eq:rrc_state_bound} by induction. Substituting the state bound into
    \cref{eq:rrc_identity} gives~\cref{eq:rrc_finite_bound}. Taking the limit
    superior as $D\rightarrow\infty$ proves~\cref{eq:risk_sizing}.
\end{proof}

%% file: utils/bibliography.bib
@inproceedings{angelopoulos2024conformal,
  author    = {Angelopoulos, Anastasios Nikolas and Bates, Stephen and Fisch, Adam and Lei, Lihua and Schuster, Tal},
  title     = {Conformal Risk Control},
  booktitle = {Proc. 12th Int. Conf. Learn. Represent. (ICLR)},
  year      = {2024},
  publisher = {OpenReview.net},
  address   = {Vienna, Austria}
}

@article{cannon2011quantile,
  author  = {Cannon, Alex J.},
  title   = {Quantile Regression Neural Networks: Implementation in {R} and Application to Precipitation Downscaling},
  journal = {Comput. Geosci.},
  volume  = {37},
  number  = {9},
  pages   = {1277--1284},
  year    = {2011},
  doi     = {10.1016/j.cageo.2010.07.005}
}

@inproceedings{chen2016xgboost,
  author    = {Chen, Tianqi and Guestrin, Carlos},
  title     = {{XGBoost}: A Scalable Tree Boosting System},
  booktitle = {Proc. 22nd ACM SIGKDD Int. Conf. Knowl. Discovery Data Mining},
  pages     = {785--794},
  year      = {2016},
  publisher = {Association for Computing Machinery},
  address   = {New York, NY, USA},
  doi       = {10.1145/2939672.2939785},
  isbn      = {9781450342322}
}

@article{feldman2023rolling,
  author  = {Feldman, Shai and Ringel, Liran and Bates, Stephen and Romano, Yaniv},
  title   = {Achieving Risk Control in Online Learning Settings},
  journal = {Trans. Mach. Learn. Res.},
  year    = {2023},
  month   = nov
}

@inproceedings{gibbs2021adaptive,
  author    = {Gibbs, Isaac and Cand{\`e}s, Emmanuel J.},
  title     = {Adaptive Conformal Inference under Distribution Shift},
  booktitle = {Adv. Neural Inf. Process. Syst.},
  volume    = {34},
  pages     = {1660--1672},
  year      = {2021},
  publisher = {Curran Associates, Inc.},
  address   = {Red Hook, NY, USA}
}

@article{holttinen2012methodologies,
  author  = {Holttinen, Hannele and Milligan, Michael and Ela, Erik and Menemenlis, Nickie and Dobschinski, Jan and Rawn, Barry and Bessa, Ricardo J. and Flynn, Damian and Gomez-Lazaro, Emilio and Detlefsen, Nina K.},
  title   = {Methodologies to Determine Operating Reserves Due to Increased Wind Power},
  journal = {IEEE Trans. Sustain. Energy},
  volume  = {3},
  number  = {4},
  pages   = {713--723},
  year    = {2012},
  doi     = {10.1109/TSTE.2012.2208207}
}

@article{meinshausen2006quantile,
  author  = {Meinshausen, Nicolai},
  title   = {Quantile Regression Forests},
  journal = {J. Mach. Learn. Res.},
  volume  = {7},
  number  = {35},
  pages   = {983--999},
  year    = {2006}
}

@book{bishop2006pattern,
  author    = {Bishop, Christopher M.},
  title     = {Pattern Recognition and Machine Learning},
  publisher = {Springer},
  address   = {New York, NY, USA},
  year      = {2006}
}

@inproceedings{ortega2020risk,
  author    = {Ortega-Vazquez, Miguel A. and Costilla-Enriquez, Napoleon and Ela, Erik and Tuohy, Aidan},
  title     = {Risk-Based Reserve Procurement},
  booktitle = {Proc. Int. Conf. Probabilistic Methods Appl. Power Syst. (PMAPS)},
  pages     = {1--7},
  year      = {2020},
  publisher = {IEEE},
  address   = {Piscataway, NJ, USA},
  doi       = {10.1109/PMAPS47429.2020.9183585}
}

@inproceedings{ortega2022generation,
  author    = {Ortega-Vazquez, Miguel A.},
  title     = {Generation of Multi-Resolution Scenarios of Stochastic Variables for Operation Planning Studies},
  booktitle = {Proc. 17th Int. Conf. Probabilistic Methods Appl. Power Syst. (PMAPS)},
  pages     = {1--6},
  year      = {2022},
  publisher = {IEEE},
  address   = {Piscataway, NJ, USA},
  doi       = {10.1109/PMAPS53380.2022.9810581}
}

@inproceedings{romano2019conformalized,
  author    = {Romano, Yaniv and Patterson, Evan and Cand{\`e}s, Emmanuel J.},
  title     = {Conformalized Quantile Regression},
  booktitle = {Adv. Neural Inf. Process. Syst.},
  volume    = {32},
  year      = {2019},
  publisher = {Curran Associates, Inc.},
  address   = {Red Hook, NY, USA}
}

@inproceedings{vovk2012conditional,
  author    = {Vovk, Vladimir},
  title     = {Conditional Validity of Inductive Conformal Predictors},
  booktitle = {Proc. Asian Conf. Mach. Learn.},
  series    = {Proc. Mach. Learn. Res.},
  volume    = {25},
  pages     = {475--490},
  year      = {2012},
  publisher = {PMLR},
  address   = {Singapore Management University, Singapore}
}

@article{devos2019dynamic,
  author  = {De Vos, K. and Stevens, N. and Devolder, O. and Papavasiliou, A. and Hebb, B. and Matthys-Donnadieu, J.},
  title   = {Dynamic Dimensioning Approach for Operating Reserves: Proof of Concept in {Belgium}},
  journal = {Energy Policy},
  volume  = {124},
  pages   = {272--285},
  year    = {2019},
  doi     = {10.1016/j.enpol.2018.09.031}
}

@article{dhaka2026physics,
  author  = {Dhaka, Poonam and Sreejeth, Mini and Tripathi, M. M.},
  title   = {Physics-Informed Adaptive Conformal Prediction for Wind Power Under Dynamic Weather Regimes},
  journal = {IEEE Trans. Ind. Appl.},
  volume  = {62},
  number  = {4},
  pages   = {6856--6869},
  year    = {2026},
  doi     = {10.1109/TIA.2026.3653910}
}

@article{kath2021conformal,
  author  = {Kath, Christopher and Ziel, Florian},
  title   = {Conformal Prediction Interval Estimation and Applications to Day-Ahead and Intraday Power Markets},
  journal = {Int. J. Forecast.},
  volume  = {37},
  number  = {2},
  pages   = {777--799},
  year    = {2021},
  doi     = {10.1016/j.ijforecast.2020.09.006}
}

@article{zhao2021operating,
  author  = {Zhao, Changfei and Wan, Can and Song, Yonghua},
  title   = {Operating Reserve Quantification Using Prediction Intervals of Wind Power: An Integrated Probabilistic Forecasting and Decision Methodology},
  journal = {IEEE Trans. Power Syst.},
  volume  = {36},
  number  = {4},
  pages   = {3701--3714},
  year    = {2021},
  doi     = {10.1109/TPWRS.2021.3053847}
}

@misc{ercotdata,
  author       = {{Electric Reliability Council of Texas}},
  title        = {{ERCOT} Market Information and Public Data},
  howpublished = {\url{https://www.ercot.com/mp/data-products}},
  note         = {Accessed: Sep. 7, 2026},
  year         = {2026}
}

@article{costilla2023operating,
  author  = {Costilla-Enriquez, Napoleon and Ortega-Vazquez, Miguel A. and Tuohy, Aidan and Motley, Amber and Webb, Rebecca},
  title   = {Operating Dynamic Reserve Dimensioning Using Probabilistic Forecasts},
  journal = {IEEE Trans. Power Syst.},
  volume  = {38},
  number  = {1},
  pages   = {603--616},
  year    = {2023},
  doi     = {10.1109/TPWRS.2022.3163106}
}

@book{kirschen2018fundamentals,
  author    = {Kirschen, Daniel S. and Strbac, Goran},
  title     = {Fundamentals of Power System Economics},
  edition   = {2},
  publisher = {John Wiley \& Sons},
  year      = {2018},
  isbn      = {9781119213253}
}

@techreport{ela2011operating,
  author      = {Ela, Erik and Milligan, Michael and Kirby, Brendan},
  title       = {Operating Reserves and Variable Generation},
  institution = {National Renewable Energy Laboratory},
  number      = {NREL/TP-5500-51978},
  year        = {2011},
  doi         = {10.2172/1023095},
  url         = {https://www.nrel.gov/docs/fy11osti/51978.pdf}
}

@techreport{nerc2025ltra,
  author      = {{North American Electric Reliability Corporation}},
  title       = {2025 Long-Term Reliability Assessment},
  institution = {North American Electric Reliability Corporation},
  address     = {Atlanta, GA, USA},
  month       = jan,
  year        = {2026},
  url         = {https://www.nerc.com/globalassets/our-work/assessments/nerc_ltra_2025.pdf}
}

@article{matos2011setting,
  author  = {Matos, Manuel A. and Bessa, Ricardo J.},
  title   = {Setting the Operating Reserve Using Probabilistic Wind Power Forecasts},
  journal = {IEEE Trans. Power Syst.},
  volume  = {26},
  number  = {2},
  pages   = {594--603},
  year    = {2011},
  doi     = {10.1109/TPWRS.2010.2065818}
}

@article{bessa2012reserve,
  author  = {Bessa, Ricardo J. and Matos, Manuel A. and Costa, Ivo C. and Bremermann, Leonardo and Franchin, Ivan Gustavo and Pestana, Rui and Machado, N{\'e}lio and Waldl, Hans-Peter and Wichmann, Christian},
  title   = {Reserve Setting and Steady-State Security Assessment Using Wind Power Uncertainty Forecast: A Case Study},
  journal = {IEEE Trans. Sustain. Energy},
  volume  = {3},
  number  = {4},
  pages   = {827--836},
  year    = {2012},
  doi     = {10.1109/TSTE.2012.2199340}
}

@article{fahiman2019dynamic,
  author  = {Fahiman, Fateme and Disano, Steven and Erfani, Sarah Monazam and Mancarella, Pierluigi and Leckie, Christopher},
  title   = {Data-Driven Dynamic Probabilistic Reserve Sizing Based on Dynamic {Bayesian} Belief Networks},
  journal = {IEEE Trans. Power Syst.},
  volume  = {34},
  number  = {3},
  pages   = {2281--2291},
  year    = {2019},
  doi     = {10.1109/TPWRS.2018.2884711}
}

@article{koenker1978regression,
  author  = {Koenker, Roger and Bassett, Jr., Gilbert},
  title   = {Regression Quantiles},
  journal = {Econometrica},
  volume  = {46},
  number  = {1},
  pages   = {33--50},
  year    = {1978},
  doi     = {10.2307/1913643}
}

@article{barber2023beyond,
  author  = {Barber, Rina Foygel and Cand{\`e}s, Emmanuel J. and Ramdas, Aaditya and Tibshirani, Ryan J.},
  title   = {Conformal Prediction Beyond Exchangeability},
  journal = {Ann. Stat.},
  volume  = {51},
  number  = {2},
  pages   = {816--845},
  year    = {2023},
  doi     = {10.1214/23-AOS2276}
}

@misc{ercotNodalProtocols2026,
  author       = {{Electric Reliability Council of Texas}},
  title        = {{ERCOT} Nodal Protocols, Section 4: Day-Ahead Operations},
  howpublished = {\url{https://www.ercot.com/files/docs/2024/06/28/04-080126_Nodal.docx}},
  note         = {Effective Aug. 1, 2026; accessed: Sep. 12, 2026},
  year         = {2026}
}

@misc{ercotASMethodology2026,
  author       = {{ERCOT}},
  title        = {Recommendation Regarding 2026 {ERCOT} Methodologies for Determining Minimum Ancillary Service Requirements},
  howpublished = {\url{https://www.ercot.com/files/docs/2025/09/15/15-Recommendation-regarding-2026-ERCOT-Methodologies-for-Determining-Minimum-Ancillary-Service-Requirements.pdf}},
  note         = {Accessed: Sep. 9, 2026},
  year         = {2025}
}

@misc{caisoBPM2026,
  author       = {{California Independent System Operator}},
  title        = {Business Practice Manual for Market Operations},
  howpublished = {\url{https://bpmcm.caiso.com/BPM\%20Document\%20Library/Market\%20Operations/BPM_for_Market\%20Operations_V104_Redline.pdf}},
  note         = {Version 104, revised May 1, 2026; accessed: Sep. 9, 2026},
  year         = {2026}
}

@misc{nyisoASManual2026,
  author       = {{New York Independent System Operator}},
  title        = {Ancillary Services Manual},
  howpublished = {\url{https://www.nyiso.com/documents/20142/2923301/ancserv.pdf}},
  note         = {Manual 2; accessed: Sep. 9, 2026},
  year         = {2026}
}

@techreport{nesoDRS2020,
  author      = {Evans, Alex and Kalair, Kieran and Warrington, Rachael and Woodhouse, Francis},
  title       = {Dynamic Reserve Setting: Improving Setting of Reserve Requirements through Machine Learning},
  institution = {National Grid Electricity System Operator},
  number      = {P20-050},
  year        = {2022},
  url         = {https://www.neso.energy/document/354026/download}
}
